\documentclass[lettersize,journal]{IEEEtran}
\IEEEoverridecommandlockouts

\usepackage[utf8]{inputenc} 
\usepackage[T1]{fontenc}
\usepackage{ifthen}
\usepackage{array}
\usepackage{mathrsfs}
\usepackage{amsthm}
\usepackage{multirow}
\usepackage{cite}
\usepackage{multicol}
\usepackage[dvipsnames]{xcolor}
\usepackage{pgfplots}
\pgfplotsset{compat=1.18}
\usepackage{comment}
\usepackage{braket}
\usepackage{stmaryrd}
\usepackage{mathtools}
\usepackage{bm}
\usepackage{amsmath,amssymb,amsfonts}
\usepackage{graphicx}
\usepackage{textcomp}
\usepackage[textsize=tiny]{todonotes}
\usepackage[hidelinks]{hyperref}
\usepackage{tikz}
\usetikzlibrary{patterns}
\usepackage{booktabs}
\usepackage[caption=false,font=footnotesize]{subfig}
\usepackage[acronym,shortcuts]{glossaries}

\def\BibTeX{{\rm B\kern-.05em{\sc i\kern-.025em b}\kern-.08em
    T\kern-.1667em\lower.7ex\hbox{E}\kern-.125emX}}
\definecolor{HP1}{RGB}{174,1,126}
\definecolor{HP2}{RGB}{222,70,155}

\definecolor{BB1}{RGB}{217,72,1}
\definecolor{BB2}{RGB}{241,105,19}
\definecolor{BB3}{RGB}{253,141,60}

\definecolor{GB1}{RGB}{33,113,181}
\definecolor{GB2}{RGB}{107,174,214}

\definecolor{QC1}{RGB}{117,107,177}

\definecolor{QD1}{RGB}{0,104,55}
\definecolor{QD2}{RGB}{35,139,69}
\definecolor{QD3}{RGB}{65,171,93}
\definecolor{QD4}{RGB}{116,216,138}

\definecolor{kit-green}{RGB}{0, 150, 130}
\colorlet{KITgreen}{kit-green}
\colorlet{kit-green100}{kit-green}
\colorlet{kit-green90}{kit-green!90!white}
\colorlet{kit-green80}{kit-green!80!white}
\colorlet{kit-green70}{kit-green!70!white}
\colorlet{kit-green60}{kit-green!60!white}
\colorlet{kit-green50}{kit-green!50!white}
\colorlet{kit-green40}{kit-green!40!white}
\colorlet{kit-green30}{kit-green!30!white}
\colorlet{kit-green25}{kit-green!25!white}
\colorlet{kit-green20}{kit-green!20!white}
\colorlet{kit-green15}{kit-green!15!white}
\colorlet{kit-green10}{kit-green!10!white}
\colorlet{kit-green5}{kit-green!5!white}

\definecolor{kit-blue}{RGB}{70, 100, 170}
\colorlet{KITblue}{kit-blue}
\colorlet{kit-blue100}{kit-blue}
\colorlet{kit-blue90}{kit-blue!90!white}
\colorlet{kit-blue80}{kit-blue!80!white}
\colorlet{kit-blue70}{kit-blue!70!white}
\colorlet{kit-blue60}{kit-blue!60!white}
\colorlet{kit-blue50}{kit-blue!50!white}
\colorlet{kit-blue40}{kit-blue!40!white}
\colorlet{kit-blue30}{kit-blue!30!white}
\colorlet{kit-blue25}{kit-blue!25!white}
\colorlet{kit-blue20}{kit-blue!20!white}
\colorlet{kit-blue15}{kit-blue!15!white}
\colorlet{kit-blue10}{kit-blue!10!white}
\colorlet{kit-blue5}{kit-blue!5!white}

\definecolor{kit-royalblue}{RGB}{0, 45, 76}
\colorlet{kit-royalblue100}{kit-royalblue}
\colorlet{kit-royalblue90}{kit-royalblue!90!white}
\colorlet{kit-royalblue80}{kit-royalblue!80!white}
\colorlet{kit-royalblue70}{kit-royalblue!70!white}
\colorlet{kit-royalblue60}{kit-royalblue!60!white}
\colorlet{kit-royalblue50}{kit-royalblue!50!white}
\colorlet{kit-royalblue40}{kit-royalblue!40!white}
\colorlet{kit-royalblue30}{kit-royalblue!30!white}
\colorlet{kit-royalblue25}{kit-royalblue!25!white}
\colorlet{kit-royalblue20}{kit-royalblue!20!white}
\colorlet{kit-royalblue15}{kit-royalblue!15!white}
\colorlet{kit-royalblue10}{kit-royalblue!10!white}
\colorlet{kit-royalblue5}{kit-royalblue!5!white}

\definecolor{kit-iceblue100}{RGB}{30, 53, 69}
\definecolor{kit-iceblue70}{RGB}{68, 94, 111}
\definecolor{kit-iceblue50}{RGB}{168, 185, 196}
\definecolor{kit-iceblue30}{RGB}{218, 225, 230}

\definecolor{kit-red}{RGB}{162, 34, 35}
\colorlet{KITred}{kit-red}
\colorlet{kit-red100}{kit-red}
\colorlet{kit-red90}{kit-red!90!white}
\colorlet{kit-red80}{kit-red!80!white}
\colorlet{kit-red70}{kit-red!70!white}
\colorlet{kit-red60}{kit-red!60!white}
\colorlet{kit-red50}{kit-red!50!white}
\colorlet{kit-red40}{kit-red!40!white}
\colorlet{kit-red30}{kit-red!30!white}
\colorlet{kit-red25}{kit-red!25!white}
\colorlet{kit-red20}{kit-red!20!white}
\colorlet{kit-red15}{kit-red!15!white}
\colorlet{kit-red10}{kit-red!10!white}
\colorlet{kit-red5}{kit-red!5!white}

\definecolor{kit-yellow}{RGB}{252, 229, 0}
\colorlet{KITyellow}{kit-yellow}
\colorlet{kit-yellow100}{kit-yellow}
\colorlet{kit-yellow90}{kit-yellow!90!white}
\colorlet{kit-yellow80}{kit-yellow!80!white}
\colorlet{kit-yellow70}{kit-yellow!70!white}
\colorlet{kit-yellow60}{kit-yellow!60!white}
\colorlet{kit-yellow50}{kit-yellow!50!white}
\colorlet{kit-yellow40}{kit-yellow!40!white}
\colorlet{kit-yellow30}{kit-yellow!30!white}
\colorlet{kit-yellow25}{kit-yellow!25!white}
\colorlet{kit-yellow20}{kit-yellow!20!white}
\colorlet{kit-yellow15}{kit-yellow!15!white}
\colorlet{kit-yellow10}{kit-yellow!10!white}
\colorlet{kit-yellow5}{kit-yellow!5!white}

\definecolor{kit-orange}{RGB}{223, 155, 27}
\colorlet{KITorange}{kit-orange}
\definecolor{kit-orange}{RGB}{223, 155, 27}
\colorlet{kit-orange100}{kit-orange}
\colorlet{kit-orange90}{kit-orange!90!white}
\colorlet{kit-orange80}{kit-orange!80!white}
\colorlet{kit-orange70}{kit-orange!70!white}
\colorlet{kit-orange60}{kit-orange!60!white}
\colorlet{kit-orange50}{kit-orange!50!white}
\colorlet{kit-orange40}{kit-orange!40!white}
\colorlet{kit-orange30}{kit-orange!30!white}
\colorlet{kit-orange25}{kit-orange!25!white}
\colorlet{kit-orange20}{kit-orange!20!white}
\colorlet{kit-orange15}{kit-orange!15!white}
\colorlet{kit-orange10}{kit-orange!10!white}
\colorlet{kit-orange5}{kit-orange!5!white}

\definecolor{kit-lightgreen}{RGB}{140, 182, 60}
\colorlet{KITlightgreen}{kit-lightgreen}
\colorlet{kit-lightgreen100}{kit-lightgreen}
\colorlet{kit-lightgreen90}{kit-lightgreen!90!white}
\colorlet{kit-lightgreen80}{kit-lightgreen!80!white}
\colorlet{kit-lightgreen70}{kit-lightgreen!70!white}
\colorlet{kit-lightgreen60}{kit-lightgreen!60!white}
\colorlet{kit-lightgreen50}{kit-lightgreen!50!white}
\colorlet{kit-lightgreen40}{kit-lightgreen!40!white}
\colorlet{kit-lightgreen30}{kit-lightgreen!30!white}
\colorlet{kit-lightgreen25}{kit-lightgreen!25!white}
\colorlet{kit-lightgreen20}{kit-lightgreen!20!white}
\colorlet{kit-lightgreen15}{kit-lightgreen!15!white}
\colorlet{kit-lightgreen10}{kit-lightgreen!10!white}
\colorlet{kit-lightgreen5}{kit-lightgreen!5!white}

\definecolor{kit-purple}{RGB}{163, 16, 124}
\colorlet{KITpurple}{kit-purple}
\colorlet{kit-purple100}{kit-purple}
\colorlet{kit-purple90}{kit-purple!90!white}
\colorlet{kit-purple80}{kit-purple!80!white}
\colorlet{kit-purple70}{kit-purple!70!white}
\colorlet{kit-purple60}{kit-purple!60!white}
\colorlet{kit-purple50}{kit-purple!50!white}
\colorlet{kit-purple40}{kit-purple!40!white}
\colorlet{kit-purple30}{kit-purple!30!white}
\colorlet{kit-purple25}{kit-purple!25!white}
\colorlet{kit-purple20}{kit-purple!20!white}
\colorlet{kit-purple15}{kit-purple!15!white}
\colorlet{kit-purple10}{kit-purple!10!white}
\colorlet{kit-purple5}{kit-purple!5!white}

\definecolor{kit-brown}{RGB}{167, 130, 46}
\colorlet{KITbrown}{kit-brown}
\colorlet{kit-brown100}{kit-brown}
\colorlet{kit-brown90}{kit-brown!90!white}
\colorlet{kit-brown80}{kit-brown!80!white}
\colorlet{kit-brown70}{kit-brown!70!white}
\colorlet{kit-brown60}{kit-brown!60!white}
\colorlet{kit-brown50}{kit-brown!50!white}
\colorlet{kit-brown40}{kit-brown!40!white}
\colorlet{kit-brown30}{kit-brown!30!white}
\colorlet{kit-brown25}{kit-brown!25!white}
\colorlet{kit-brown20}{kit-brown!20!white}
\colorlet{kit-brown15}{kit-brown!15!white}
\colorlet{kit-brown10}{kit-brown!10!white}
\colorlet{kit-brown5}{kit-brown!5!white}

\definecolor{kit-cyan}{RGB}{35, 161, 224}
\colorlet{KITcyan}{kit-cyan}
\colorlet{KITcyanblue}{kit-cyan}
\colorlet{kit-cyan100}{kit-cyan}
\colorlet{kit-cyan90}{kit-cyan!90!white}
\colorlet{kit-cyan80}{kit-cyan!80!white}
\colorlet{kit-cyan70}{kit-cyan!70!white}
\colorlet{kit-cyan60}{kit-cyan!60!white}
\colorlet{kit-cyan50}{kit-cyan!50!white}
\colorlet{kit-cyan40}{kit-cyan!40!white}
\colorlet{kit-cyan30}{kit-cyan!30!white}
\colorlet{kit-cyan25}{kit-cyan!25!white}
\colorlet{kit-cyan20}{kit-cyan!20!white}
\colorlet{kit-cyan15}{kit-cyan!15!white}
\colorlet{kit-cyan10}{kit-cyan!10!white}
\colorlet{kit-cyan5}{kit-cyan!5!white}

\definecolor{kit-gray}{RGB}{0, 0, 0}
\colorlet{KITgray}{kit-gray}
\colorlet{kit-gray100}{kit-gray}
\colorlet{kit-gray90}{kit-gray!90!white}
\colorlet{kit-gray80}{kit-gray!80!white}
\colorlet{kit-gray70}{kit-gray!70!white}
\colorlet{kit-gray60}{kit-gray!60!white}
\colorlet{kit-gray50}{kit-gray!50!white}
\colorlet{kit-gray40}{kit-gray!40!white}
\colorlet{kit-gray30}{kit-gray!30!white}
\colorlet{kit-gray25}{kit-gray!25!white}
\colorlet{kit-gray20}{kit-gray!20!white}
\colorlet{kit-gray15}{kit-gray!15!white}
\colorlet{kit-gray10}{kit-gray!10!white}
\colorlet{kit-gray5}{kit-gray!5!white}

\definecolor{KITpalegreen}{RGB}{130,190,60}
\colorlet{kit-maigreen100}{KITpalegreen}
\colorlet{kit-maigreen70}{KITpalegreen!70}
\colorlet{kit-maigreen50}{KITpalegreen!50}
\colorlet{kit-maigreen30}{KITpalegreen!30}
\colorlet{kit-maigreen15}{KITpalegreen!15}

\newcommand{\sym}[1]{#1}         
\newcommand{\F}{\mathbb{F}}
\newcommand{\Fl}{\F_{2^\ell}}
\newcommand{\Flx}{\F^{\times}_{2^\ell}}
\newcommand{\HX}{\mathbf{H}_{\mathrm{X}}}
\newcommand{\HZ}{\mathbf{H}_{\mathrm{Z}}}
\newcommand{\PX}{\mathbf{P}_{\mathrm{X}}}
\newcommand{\PZ}{\mathbf{P}_{\mathrm{Z}}}
\newcommand{\supp}{\mathrm{supp}}
\newcommand{\wt}{\mathrm{wt}}

\newcommand{\rk}{\mathrm{rank}}
\newcommand{\6}{\mathbf }

\usepackage{tikz}
\usetikzlibrary{arrows,automata, positioning}
\usetikzlibrary{calc, chains,
                fit,
                positioning,
                shapes}
\usetikzlibrary{patterns}       

\usepackage{pgfplots}
\pgfplotsset{compat=newest} 

\usetikzlibrary{matrix,fadings,decorations.pathmorphing,intersections}
\usepgfplotslibrary{fillbetween}
\pgfdeclarelayer{bg}
\pgfsetlayers{bg,main}
\makeatletter
\tikzset{use path/.code=\tikz@addmode{\pgfsyssoftpath@setcurrentpath#1}}

\definecolor{myblue}{RGB}{0,0,205}
\definecolor{mygreen}{RGB}{80,160,80}
\definecolor{myred}{RGB}{178,34,34}
\definecolor{mygray}{RGB}{211,211,211}
\definecolor{myyellow}{rgb}{1.0, 0.75, 0.0}
\definecolor{mygreen2}{rgb}{0.4, 0.69, 0.2}
\definecolor{myred2}{rgb}{0.89, 0.15, 0.21}
\definecolor{myblue2}{rgb}{0.0, 0.47, 0.75}

\theoremstyle{plain}
\newtheorem{The}{Theorem}

\newtheorem{Lem}{Lemma}
\newtheorem{Cor}{Corollary}

\newtheorem{Def}{Definition}

\usetikzlibrary{patterns}
\usepackage{scalerel}

\makeatletter
\renewcommand*\env@matrix[1][*\c@MaxMatrixCols c]{%
  \hskip -\arraycolsep
  \let\@ifnextchar\new@ifnextchar
  \array{#1}}
\makeatother

\usepackage[ruled, vlined, linesnumbered]{algorithm2e}

\SetCommentSty{mycommfont}

\newacronym{AWGN}{AWGN}{additive white Gaussian noise}
\newacronym{BER}{BER}{bit error rate}
\newacronym{BP}{BP}{belief propagation}
\newacronym{CPM}{CPM}{circulant permutation matrix}
\newacronym{BP2}{BP$2$}{binary belief propagation}
\newacronym{GB}{GB}{generalized bicycle}
\newacronym{HP}{HP}{hypergraph product}
\newacronym{MS}{MS}{Min-Sum}
\newacronym{BP+OSD}{BP$2+$OSD}{belief propagation plus ordered-statistics decoding}
\newacronym{BPSK}{BPSK}{binary phase-shift keying}
\newacronym{CAMEL}{CAMEL}{Cycle Assembling and Mitigating with EnsembLe decoding}
\newacronym{CSS}{CSS}{Calderbank--Shor--Steane}
\newacronym{SP}{SP}{sum-product}
\newacronym{DPM}{DPM}{dyadic permutation matrix}
\newacronym{LDPC}{LDPC}{low-density parity-check}
\newacronym{LLR}{LLR}{log-likelihood ratio}
\newacronym{QD}{QD}{quasi-dyadic}
\newacronym{QD-LDPC}{QD-LDPC}{quasi-dyadic low-density parity-check}
\newacronym{QLDPC}{QLDPC}{quantum low-density parity-check}
\newacronym{SNR}{SNR}{signal-to-noise ratio}
\newacronym{SC-LDPC}{SC-LDPC}{spatially coupled low-density parity-check}
\newacronym{TI}{TI}{time-invariant}
\newacronym
  [longplural={parity-check matrices},
   shortplural={PCMs}]
  {PCM}{PCM}{parity-check matrix}
\newacronym{QC}{QC}{quasi-cyclic}
\newacronym{QSC}{QSC}{quantum stabilizer code}
\newacronym{QEC}{QEC}{quantum error correction}
\newacronym{QECC}{QECC}{quantum error correcting code}
\newacronym{LER}{LER}{logical error rate}
\newacronym{IS}{IS}{intersecting subset}
\newacronym{BP4}{BP$4$}{quaternary belief propagation}
\newacronym{BB}{BB}{bivariate bicycle}
\newacronym{DC}{DC}{dual-containing}
\newacronym{VN}{VN}{variable node}
\newacronym{CN}{CN}{check node}

\begin{document}

\title{High-Rate Quasi-Dyadic Quantum LDPC Codes
\thanks{
The work of Alessio Baldelli was partially supported by Agenzia per la Cybersicurezza Nazionale (ACN) under the programme for promotion of XL cycle PhD research in cybersecurity (CUP I32B24001750005). 
Parts of this work were funded by the Deutsche Forschungsgemeinschaft (DFG, German
Research Foundation) –  Project 567557088.

Alessio Baldelli, Marco Baldi, and Massimo Battaglioni are with the Department of Information Engineering, Università Politecnica delle Marche, 60131 Ancona, Italy.
Sisi Miao and Laurent Schmalen are with Communications Engineering Lab (CEL), Karlsruhe Institute of Technology (KIT), 76187 Karlsruhe, Germany.

}
}

\author{
\IEEEauthorblockN{
    Alessio Baldelli,
    Sisi Miao, 
    Laurent Schmalen,
    Marco Baldi,
    Massimo Battaglioni \\
}
}

\maketitle

\begin{abstract}
    In this work, we propose a novel design of high-rate Calderbank--Shor--Steane (CSS) quantum low-density parity-check (QLDPC) codes based on quasi-dyadic matrices. The Tanner graphs associated with the two component classical codes used to construct the proposed CSS codes are characterized by a girth of at least $6$. 
    We also derive an exact count of length-$4$ cycles in the quaternary representation of the parity-check matrix.  Monte Carlo simulations under code-capacity and phenomenological noise models show that the proposed codes have competitive finite-length logical error rate performance compared with  state-of-the-art  codes in both these settings. 
\end{abstract}

\begin{IEEEkeywords}
    CSS codes, QLDPC codes, quasi-dyadic codes
\end{IEEEkeywords}

\section{Introduction}
\label{intro}

Quantum computers are highly susceptible to noise, making \acp{QECC} necessary for reliable and fault-tolerant quantum computation. Stabilizer codes~\cite{gottesman1997stabilizercodesquantumerror} are the most commonly studied \acp{QECC}.   
The most widely used stabilizer codes are the  \ac{CSS} codes~\cite{CSS, Steane_CSS}, where a \ac{QECC} is constructed from two classical component codes whose \acp{PCM} $\6H_{\mathrm{X}}$ and $\6H_{\mathrm{Z}}$ satisfy the \emph{orthogonality condition} $\6H_{\mathrm{X}}\6H_{\mathrm{Z}}^{\top}=\60$.

\Ac{QLDPC} codes have attracted considerable attention~\cite{sparse_quantum_MacKay, Hagiwara2007, Tillich_HGP, kovalev_Pryad_GBcodes, Panteleev_degenerate, Bravyi2024_BBcodes_Nature,kenta_CPM_2026}, also due to the higher encoding rates achieved by many modern \ac{QLDPC} constructions compared to topological codes~\cite{11_kitaev1997quantum, 15_freedman2001projective}. 
A challenge in designing \ac{CSS} \ac{QLDPC} codes is to preserve the orthogonality of $\HX$ and $\HZ$, while avoiding the length-$4$ cycles in their Tanner graphs, which are  detrimental to \ac{BP}-based decoding. 
Structured constructions satisfying these requirements include \ac{QC}~\cite{Hagiwara2007}, \ac{BB}~\cite{Bravyi2024_BBcodes_Nature}, \ac{CPM}-based~\cite{kenta_CPM_2026} codes and some instances of \ac{HP}~\cite{Tillich_HGP} and  \ac{GB}~\cite{kovalev_Pryad_GBcodes} codes.  
\ac{CSS} \ac{QLDPC} codes based on dyadic and \ac{QD} matrices were introduced in~\cite{dyadics}. 
Other \ac{DC} constructions have been proposed and further studied in~\cite{baldelli2025ISTC, baldelli2026DC, kirsten_QD}. 
In the \ac{DC} setting, we have $\HX=\HZ$, leading to unavoidable length-$4$ cycles~\cite{baldelli2026DC}. To the best of our knowledge, the only previous non-\ac{DC} \ac{CSS} \ac{QD} construction is based on the \ac{CAMEL} framework~\cite{baldelli2026ISIT}. 
In this setting, the component \acp{PCM} $\HX$ and $\HZ$ have girth $4$, but all length-$4$ cycles are confined to a single column of the \ac{PCM}. The resulting \ac{CSS} codes need a specialized ensemble decoder, which is not applicable to general \ac{CSS} codes.

In this work, we propose a novel \ac{QD} \ac{CSS} construction based on an affine-Frobenius permutation technique, which yields mutually orthogonal component \acp{PCM} whose Tanner graphs have girth at least $6$. 
We further derive an exact expression for the number of length-$4$ cycles in the quaternary representation of the \ac{PCM} of the \ac{CSS} code. 
The proposed codes achieve high rates due to the large rank deficiency of their component \ac{QD} \acp{PCM}, whose redundant syndrome constraints are beneficial under phenomenological noise. Monte Carlo simulations under code-capacity and phenomenological noise show competitive finite-length \ac{LER} performance against state-of-the-art \ac{QLDPC} code instances with comparable length, rate, and stabilizer generator weight.

\section{Notation and Preliminaries}
\label{sec:preli}

The set of integers between $a$ and $b$ (endpoints included) is denoted as $[a,b]$.
We denote row vectors with bold lowercase letters, e.g., $\6a$,  while we use bold uppercase letters, e.g., $\6A$, for matrices.  The all-zero vector (matrix) is denoted by $\60$. The Hamming weight (simply weight) of a binary vector is denoted as $\wt(\6a)$, while its support is $\supp(\6a)$. 

We denote by $\mathbb{N}$ the set of natural numbers, by $\F_2$ the binary field, and by $\Fl$ the finite field of order $2^\ell$. We further let $\Flx\coloneq\Fl\setminus\{0\}$ and denote by $\alpha$ a primitive element of $\Fl$. The operator $\oplus$ denotes the bitwise XOR of the binary representations of two integers. We define the \emph{fiber} of a map $\nu: \Fl \to \Fl$ and a target scalar $\sigma \in \Fl$ as the set of elements that produce $\sigma$, i.e.,
\begin{equation*}
    \nu^{-1}(\sigma) \coloneq \{ \lambda \in \Fl \, | \, \nu(\lambda) = \sigma \}. 
\end{equation*}

\subsection{Classical and Quantum Stabilizer Codes}
\label{subsec:cla_qua_codes}

A binary linear code $C[n,k,d]$ can be represented by a \ac{PCM} $\6H\in\F_2^{m\times n}$ such that $C := \ker(\6H)$. 
We define its rank deficiency as $r_{\mathrm d}\coloneq m-\rk(\6H)$.

\Acp{QSC} are the quantum counterparts of classical binary linear codes.
Within the framework of \acp{QSC}, the most studied subclass is that of \ac{CSS} codes.

\begin{Def}[\ac{CSS} Code]
Let $\HX\in\F_2^{m_{\mathrm X}\times n}$ and $\HZ\in\F_2^{m_{\mathrm Z}\times n}$ be the \acp{PCM} of two binary linear codes satisfying
\[
\HX\HZ^\top=\60.
\]
They define a \ac{CSS} code $\mathcal{C}\llbracket n,k,d\rrbracket$ with $k=n-\rk(\HX)-\rk(\HZ)$
encoded qubits, whose binary \ac{PCM} is
\begin{equation}
\label{eq:PCM_CSS}
    \6H \coloneq
    \begin{pmatrix}[c|c]
        \6H_{\mathrm{X}} & \60\\
        \60 & \6H_{\mathrm{Z}}
    \end{pmatrix}
\in \mathbb{F}_2^{(m_{\mathrm{X}} + m_{\mathrm{Z}}) \times 2n}.
\end{equation}
\end{Def}

The \ac{CSS} \ac{PCM} can be equivalently represented over 
$\mathbb{F}_4=\{0,1,\omega,\bar{\omega} \}$  
as
\begin{equation} 
\label{eq:CSS_codes_4}
    \6H \coloneq
    \begin{pmatrix}
        \omega\6H_{\mathrm{X}}\\
        \bar{\omega}\6H_{\mathrm{Z}}
    \end{pmatrix}
    \in \mathbb{F}_4^{(m_{\mathrm{X}} + m_{\mathrm{Z}}) \times n}.
\end{equation}
A sparse \ac{PCM} defines a \ac{LDPC} code and its associated Tanner graph. We denote by $g_{\mathrm X}$ and $g_{\mathrm Z}$ the girths of the Tanner graphs associated with $\HX$ and $\HZ$, respectively, and by $\delta_{\mathrm r}$ and $\delta_{\mathrm c}$ their average check and variable degrees.  The Tanner graph associated with~\eqref{eq:CSS_codes_4} necessarily contains length-$4$ cycles; their number is denoted as $n_{4}$.

\subsection{Dyadic and Quasi-Dyadic Matrices and Codes}
\label{subsec:dyadics}

\begin{Def} [Ring of Dyadic Matrices]
    Consider $\ell \in \mathbb{N}$. 
    We define the ring of dyadic matrices $\mathcal M_{\ell}(\mathbb F_2)$
    as the set of $2^\ell\times 2^\ell$ binary matrices structured as follows
    \begin{align*}
        \6M = 
        \begin{pmatrix}
        \6A & \6B \\
        \6B & \6A
        \end{pmatrix}, \quad \6A,\6B\in\mathcal M_{\ell-1}(\mathbb F_2).
    \end{align*}
    For $\ell = 0$, we have $\mathcal M_0(\mathbb F_2) := \mathbb{F}_2$.
\end{Def}

Any dyadic matrix $\mathbf{M} \in \mathcal{M}_{\ell}(\F_2)$ is fully determined by its first row $\6m = (m_{0,0}, \ldots, m_{0,2^{\ell}-1})$. 
In fact, for any $i, j \in [0, 2^{\ell}-1]$, it holds that $m_{i,j} = m_{0,\, i \oplus j}$~\cite{baldelli2026DC}. 
We call \emph{\ac{DPM}} any dyadic matrix whose first row has weight $1$.
The set of all the \acp{DPM} is ${\mathcal D_\ell(\mathbb F_2) = \{ \6D^{(0)}\!\coloneq \6I_{2^\ell\times2^\ell}, \6D^{(1)},\ldots, \6D^{(2^\ell-1)} \}\subseteq \mathcal M_\ell(\mathbb F_2)}$. The $r$-th row of  $\mathbf{D}^{(l)}$ has its unique $1$ in the column indexed by $r \oplus l$. 
In particular, each $\6D^{(l)} \in \mathcal{D}_\ell(\mathbb{F}_2)$ has the unique non-zero entry of the first row in the $l$-th position. 
A matrix is called \ac{QD} if it is composed only of dyadic blocks and a classical linear code is called \ac{QD} if it admits a generator matrix or a \ac{PCM} that is \ac{QD}. We represent each \ac{QD} matrix composed only of \acp{DPM} by an \emph{exponent matrix} $\6P$.

\begin{Def}[Exponent Matrix]
    Let $\6P \in \mathbb{F}_{2^\ell}^{w \times N}$. Each entry $p_{i,j} \in \mathbb{F}_{2^\ell}$ identifies the \ac{DPM} \(\6D^{(\psi(p_{i,j}))}\),  where ${\psi : \mathbb{F}_{2^\ell} \to [0,2^\ell-1]}$ is the bijection defined by
    \[
    \psi\!\left(\sum_{k=0}^{\ell-1} c_k \alpha^k\right) = \sum_{k=0}^{\ell-1} c_k 2^k,
    \qquad c_k \in \mathbb{F}_2.
    \]
    We further denote by $\lambda_t \coloneqq \psi^{-1}(t) \in \mathbb{F}_{2^\ell}$, $t\in[0,2^\ell-1]$, the unique field element associated with the index $t$.
\end{Def}
Since  $\psi(\beta + \gamma) = \psi(\beta) \oplus \psi(\gamma)$, the $r$-th row of  $\mathbf{D}^{(\psi(p_{i, j}))}$ has its unique non-zero element in column $\psi(p_{i, j}) \oplus r = \psi(p_{i, j} + \lambda_r)$. \emph{Lifting} replaces each entry of $\6P$ with the corresponding \ac{DPM}, to obtain the \ac{PCM} $\6H$. We define the $i$-th \emph{block-row} as the sequence of blocks in $\6H$ associated with $p_{i,0},\ldots,p_{i,N-1}$, and the $j$-th \emph{block-column} as the sequence of blocks in $\6H$ associated with $p_{0,j},\ldots,p_{w-1,j}$. Moreover, we define the \emph{Frobenius map} $\mu: \Fl \to \Fl$ as $\mu(\lambda) = \lambda^2$. Note that $\mu$ is linear in $\F_2$ and  is a bijection.

\subsection{Noise Models}
\label{subsec:noise_models}
We evaluate the proposed \ac{CSS} codes under code-capacity and phenomenological noise. Under code-capacity depolarizing noise, syndrome measurements are perfect, and each data qubit independently undergoes an $\sym{X}$, $\sym{Y}$, or $\sym{Z}$ error with probability $\epsilon/3$. Under phenomenological noise, syndrome readout errors additionally occur independently with probability $p$. We consider the one-round single-shot setting of~\cite[Sec.~5]{Ostrev2024classicalproduct}, where a potentially redundant set of syndromes is measured once and the linear dependencies among the rows of the \ac{PCM} provide additional parity constraints on the measured syndrome. For the $\sym{X}$ component, the binary decoder uses the extended \ac{PCM}
\[
\widetilde{\6H}_{\mathrm{X}}
\coloneq
\begin{pmatrix}[c|c]
    \HX & \6I\\
    \60 & \6L_{\mathrm{X}}
\end{pmatrix},
\]
where $\6L_{\mathrm{X}}\in\F_2^{r_{\mathrm{d},\mathrm{X}}\times m_{\mathrm{X}}}$ is a full-row-rank \emph{meta-check matrix} satisfying
\(\ker(\6L_{\mathrm{X}})=\operatorname{row}(\HX^\top)
\),  and hence $\6L_{\mathrm{X}}\HX=\60$. 
If $r_{\mathrm{d},\mathrm{X}}=0$, $\6L_{\mathrm{X}}$ is empty. 
The minimum distance $d_{\text{m}, \mathrm{X}}$ of the classical code $\ker(\6L_{\mathrm{X}})$, referred to as the \emph{meta-check distance}, characterizes its capability to correct syndrome readout errors, and is upper bounded by the minimum column weight of $\HX$~\cite{Ostrev2024classicalproduct}. 
The same construction  applies to $\HZ$.
The extended \acp{PCM}  are then used to construct the extended \ac{PCM} for \ac{BP4} in the following form
\begin{equation}
    \widetilde{\6H}
\coloneq
\left(
\begin{array}{c|cc}
\omega \HX
    & \6I_{m_{\mathrm X}}
    & \60_{m_{\mathrm X}\times m_{\mathrm Z}} \\[1mm]
\60_{r_{\mathrm{d,X}}\times n}
    & \6L_{\mathrm X}
    & \60_{r_{\mathrm{d,X}}\times m_{\mathrm Z}} \\ 
\bar{\omega}\HZ
    & \60_{m_{\mathrm Z}\times m_{\mathrm X}}
    & \6I_{m_{\mathrm Z}} \\[1mm]
\60_{r_{\mathrm{d,Z}}\times n}
    & \60_{r_{\mathrm{d,Z}}\times m_{\mathrm X}}
    & \6L_{\mathrm Z}
\end{array}
\right).
\label{eq:Htilde}
\end{equation}
Eq. \eqref{eq:Htilde} defines a hybrid Tanner graph. The variable nodes associated with its left block represent quaternary Pauli errors and employ the \ac{BP4} update rule, whereas those associated with its right block represent binary syndrome readout errors and use the standard binary log-likelihood-ratio update~\cite{Ostrev2024classicalproduct}.

\section{Code Design}
\label{sec:design}

The parameter $\ell \in \mathbb{N}$  specifies both the field $\Fl$ and the size of the \acp{DPM}.  
Moreover, let \mbox{$N \coloneq 2^{\ell}$}, $w_{\mathrm{X}}, w_{\mathrm{Z}} \in [2, 2^{\ell}-1]$.
We define two exponent matrices $\PX \in \Fl^{w_{\mathrm{X}} \times N}$ and $\PZ \in \Fl^{w_{\mathrm{Z}} \times N}$ whose rows are given by the following permutations of $\Fl$:
\begin{equation}
\label{eq:aff_perm_XZ}
    p^{(\mathrm{X})}_{u,j} = a_u \lambda_j + b_u,
    \qquad\qquad
    p^{(\mathrm{Z})}_{v,j} = c_v \lambda_j^2 + d_v,
\end{equation}
where $u \in [0, w_{\mathrm{X}}-1]$, $v \in [0, w_{\mathrm{Z}}-1]$, $j \in [0, 2^{\ell}-1]$, and $\lambda_j=\psi^{-1}(j)$. The multipliers $a_u\in\Flx$ are pairwise distinct, as are the multipliers $c_v\in\Flx$, while $b_u,d_v\in\Fl$ are arbitrary.
No condition is imposed between the two sets of multipliers. We  refer to \eqref{eq:aff_perm_XZ} as the \emph{affine-Frobenius} construction.

\subsection{Girth of the Component Tanner Graphs}
We show next that the construction in~\eqref{eq:aff_perm_XZ} leads to a pair of exponent matrices $\PX$ and $\PZ$ defining two \ac{QD} codes whose Tanner graphs have girth at least $6$.

\begin{The}[Girth condition]
\label{theo:girth}
    Let $\PX \in \Fl^{w_{\mathrm{X}} \times N}$ and $\PZ \in \Fl^{w_{\mathrm{Z}} \times N}$ be two exponent matrices obtained with the affine-Frobenius construction of~\eqref{eq:aff_perm_XZ}.
    After lifting with \acp{DPM} of size $2^{\ell}$, we get the two \acp{PCM} $\HX \in \mathbb F_2^{w_{\mathrm{X}}2^{\ell} \times N2^{\ell}}$ and $\HZ \in \mathbb F_2^{w_{\mathrm{Z}}2^{\ell} \times N2^{\ell}}$
    Then,  the Tanner graphs of  both $\HX$ and $\HZ$ have girth at least $6$.
\end{The}
\begin{IEEEproof}
    The result for $\HX$  follows directly from~\cite[Theorem~1]{baldelli2026ISIT}. Indeed, for any $u \neq u'$, 
    \[
    j \mapsto p^{(\mathrm{X})}_{u, j} + p^{(\mathrm{X})}_{u', j} = (a_u + a_{u'})\lambda_j + (b_u + b_{u'})
    \]
    is a bijection from $[0,2^\ell-1]$ to $\Fl$, since $a_u+a_{u'}\neq 0$. 
    Hence, all values $p^{(\mathrm{X})}_{u,j} + p^{(\mathrm{X})}_{u',j}$ are distinct for distinct $j$.  For $\HZ$, for any $v\neq v'$, 
    \[
    j \mapsto p^{(\mathrm{Z})}_{v, j} + p^{(\mathrm{Z})}_{v', j} = (c_v+c_{v'})\lambda_j^2+(d_v+d_{v'})
    \]
    is a bijection from $[0,2^\ell-1]$ to  $\Fl$, since $c_v+c_{v'}\neq0$ and $\mu$ is bijective. Hence, all values  $p^{(\mathrm Z)}_{v,j}+p^{(\mathrm Z)}_{v',j}$ are distinct for distinct $j$, and~\cite[Theorem~1]{baldelli2026ISIT} excludes length-$4$ cycles.
\end{IEEEproof}

\subsection{Orthogonality}
\label{subsec:ortho}

Let us show that $\HX$ and $\HZ$ are mutually orthogonal.

\begin{Lem}[Overlap parity]
\label{theo:overlap}
    Fix the $u$-th block-row of $\HX$ and the $v$-th block-row of $\HZ$.  Let $\6u_r$ denote the $r$-th row within the former and $\6v_s$ the $s$-th row within the latter, with $r,s\in[0,2^\ell-1]$. Then
    \[
    |\supp(\6u_r)\cap\supp(\6v_s)|\in\{0,2\}.
    \]
\end{Lem}
\begin{IEEEproof}
    Fix a pair of rows $(\6u_r,\6v_s)$ as in the statement and define $\kappa\coloneq a_u/c_v\in\Flx$.
    In the block-column $j\in[0,2^\ell-1]$, $\6u_r$ has its unique $1$ at column $\psi\!\left(p^{(\mathrm{X})}_{u,j}+\lambda_r\right)$, while $\6v_s$ has its unique $1$ at column $\psi\!\left(p^{(\mathrm{Z})}_{v,j}+\lambda_s\right)$. Since $\psi$ is a bijection, the two rows overlap in the $j$-th block-column iff
    \begin{align*}
        p^{(\mathrm{X})}_{u,j}+p^{(\mathrm{Z})}_{v,j} &= \lambda_r+\lambda_s,\\
        a_u\lambda_j+c_v\lambda_j^2+(b_u+d_v) &= \lambda_r+\lambda_s.
    \end{align*}
    Then, define the map 
    $\Delta:\Fl\to\Fl$ by 
    \[
    \Delta(\lambda)\coloneq a_u\lambda+c_v\lambda^2+(b_u+d_v),
    \]
    whose linear part is
    \[
    \mathcal{L}(\lambda)\coloneq a_u\lambda+c_v\lambda^2=\lambda(a_u+c_v\lambda).
    \]
    Since $a_u,c_v\in\Flx$, we have $ \ker(\mathcal{L})=\{0,\kappa\}$.    Now, consider a target value $\sigma\in\Fl$. If $\sigma\notin\operatorname{Im}(\Delta)$, then $\Delta^{-1}(\sigma)=\varnothing$. Otherwise, let $\lambda_0\in\Fl$ be such that $\Delta(\lambda_0)=\sigma$. An element $\lambda\in\Fl$ belongs to the same fiber iff $\Delta(\lambda)=\Delta(\lambda_0)$. Since $\Delta(\lambda)=\mathcal{L}(\lambda)+(b_u+d_v)$ and $\mathcal{L}$ is $\F_2$-linear, we have
    \begin{align*}
    \Delta(\lambda)=\Delta(\lambda_0)
    &\iff \mathcal{L}(\lambda)=\mathcal{L}(\lambda_0)\\
    &\iff \lambda+\lambda_0\in\ker(\mathcal{L}).
    \end{align*}
    Therefore, the only elements belonging to the fiber are \(\lambda=\lambda_0\) or \(\lambda=\lambda_0+\kappa\). 
    Since $\kappa\neq0$,  every non-empty fiber of $\Delta$ can be written as
    \(
    \Delta^{-1}(\sigma)=\{\lambda_0,\lambda_0+\kappa\}
    \). 
    For the fixed pair $(\6u_r,\6v_s)$, the overlapping block-columns are determined by the solutions of
    \(
    \Delta(\lambda_j)=\lambda_r+\lambda_s
    \). 
    Since $j\mapsto\lambda_j$ is a bijection, this equation has either zero or two solutions in $j$. Each solution corresponds to exactly one common non-zero position, since each row has a unique $1$ in each block-column. Consequently,
    \[
    |\supp(\6u_r)\cap\supp(\6v_s)|\in\{0,2\}.
    \]\end{IEEEproof}

\begin{The}[CSS orthogonality]
\label{theo:orthogonality}
Let $\PX$ and $\PZ$ be designed as in~\eqref{eq:aff_perm_XZ}, and let $\HX$ and $\HZ$ be the associated  \acp{PCM}.
Then,  
\begin{equation*}
    \HX \HZ^{\top} = \60
\end{equation*}
for every choice of the constants $b_u,$ $d_v$. 
\end{The}
\begin{IEEEproof}
    The proof simply follows from Lemma~\ref{theo:overlap}.
\end{IEEEproof}

\subsection{Cycle Analysis}
\label{subsec:cycles}
Let us count the number of length-$4$ cycles in the quaternary representation of the \ac{PCM} $\6H$  of the proposed codes.

\begin{Lem}[Counting Overlapping Row Pairs]
\label{lem:counting_couples}
    Let $\HX$ and $\HZ$ be designed using the affine-Frobenius construction \eqref{eq:aff_perm_XZ} and lifted with \acp{DPM} of size $2^{\ell}\times2^{\ell}$. Then, for every pair of block-rows, one from $\HX$ and one from $\HZ$, exactly $2^{2\ell-1}$ of the $2^{2\ell}$ row pairs $(\6u_r,\6v_s)$ satisfy
    \[
    |\supp(\6u_r)\cap\supp(\6v_s)|=2,
    \]
    while all the remaining pairs satisfy
    \[
    |\supp(\6u_r)\cap\supp(\6v_s)|=0.
    \]
\end{Lem}
\begin{IEEEproof}
By Lemma~\ref{theo:overlap}, every row pair $(\6u_r,\6v_s)$ overlaps in either $0$ or $2$ positions, with overlap $2$ iff   
    \(
    r\oplus s\in\psi(\operatorname{Im}(\Delta))
    \), since $\lambda_r+\lambda_s=\lambda_{r\oplus s}$.
    By the rank-nullity theorem applied to the $\F_2$-linear part $\mathcal{L}$ of $\Delta$, we have
    \[
    \dim(\operatorname{Im}(\mathcal{L}))
    =
    \dim(\Fl)-\dim(\ker(\mathcal{L}))
    =
    \ell-1,
    \]
    since $\ker(\mathcal{L})=\{0,\kappa\}$ has dimension $1$ over $\F_2$. Moreover,
    \[
    \operatorname{Im}(\Delta)
    =
    (b_u+d_v)+\operatorname{Im}(\mathcal{L}),
    \]
    so that
    \(
    |\operatorname{Im}(\Delta)|=2^{\ell-1}
    \).
    For fixed $r$, the map $s\mapsto r\oplus s$ is a bijection, so exactly $2^{\ell-1}$ values of $s$ satisfy $r\oplus s\in\psi(\operatorname{Im}(\Delta))$. 
    Summing over the $2^{\ell}$ possible values of $r$ gives
    \(
    2^{2\ell-1}
    \)
    row pairs with exactly two overlapping positions.
\end{IEEEproof}

\begin{Cor}[Number of length-$4$ cycles]
   
    Let $\6H$ be the quaternary \ac{PCM} associated with $\HX$ and $\HZ$ obtained from~\eqref{eq:aff_perm_XZ} with lifting size $2^{\ell}$. Then, the number of length-$4$ cycles in its Tanner graph is 
    \begin{equation*}
        n_4 = w_{\mathrm{X}} w_{\mathrm{Z}} 2^{2\ell - 1}.
    \end{equation*}
\end{Cor}
\begin{IEEEproof}
By Theorem~\ref{theo:girth}, length-$4$ cycles cannot involve two rows both belonging to $\HX$ or both to $\HZ$. There are $w_{\mathrm X}w_{\mathrm Z}$ pairs of block-rows, and by Lemma~\ref{lem:counting_couples}, each contains exactly $2^{2\ell-1}$ row pairs overlapping in two positions. Each such pair determines one length-$4$ cycle, which gives the result.
\end{IEEEproof}

\section{Numerical Results}
\label{sec:num_res}

In this section, we show the finite-length \ac{LER} performance of the proposed  codes by Monte Carlo simulations, under both code-capacity and phenomenological noise, using a \ac{BP4} decoder as in~\cite{MiaoQuat}, running at most $50$ flooding-schedule iterations. For each value of the depolarizing probability $\epsilon$, the simulation stops after $100$ logical errors have been found.

We consider four high-rate \ac{QD} \ac{QLDPC} codes, constructed from \eqref{eq:aff_perm_XZ}. 
Their parameters are available in Tab.~\ref{tab:QD_parameters}. Their stabilizer generator weight is $\delta_{\text{r}} = N$, while the column weight $\delta_{\text{c}}$ of  $\6H_{\mathrm{X}}$ and $\6H_{\mathrm{Z}}$ is $w_{\mathrm{X}}$ and $w_{\mathrm{Z}}$, respectively.  
For all the \ac{QD} codes of Tab.~\ref{tab:QD_parameters}, we choose $a_{u}  = \alpha^{u}$, $c_{v}=\alpha^v$, and $b_{u} = d_{v} = 0$. 

\begin{table}[tb!]
\centering
\caption{Parameters of the proposed \ac{CSS} \ac{QD} \ac{QLDPC} codes and of benchmark \ac{CSS} \ac{QLDPC} codes.}
\label{tab:QD_parameters}
\label{tab:codes_parameters}
\footnotesize
\setlength{\tabcolsep}{2pt}
\resizebox{\columnwidth}{!}{
\begin{tabular}{c c c c c c c c c}
\toprule
\textbf{Code} & \multicolumn{7}{c}{\textbf{Parameters}} & \textbf{Ref.} \\
\cline{2-8}
& $\delta_{\mathrm{r}}$ & $\delta_{\mathrm{c}}$
& $r_{\mathrm{d,X}}/r_{\mathrm{d,Z}}$
& $d_{\mathrm{m,X}}/d_{\mathrm{m,Z}}$
& $g_{\mathrm{X}}/g_{\mathrm{Z}}$
& $n_4$
& $R$ & \\
\hline
$\mathcal{C}_{\mathrm{QD}1}\llbracket 64, 12, 8 \rrbracket$
& $8$ & $7/7$ & $30/30$ & $7/7$ & $6/6$ & $1\,568$ & $0.19$ & \\
$\mathcal{C}_{\mathrm{QD}2}\llbracket 64, 18, 8 \rrbracket$
& $8$ & $4/4$ & $9/9$ & $4/4$ & $6/6$ & $512$ & $0.28$ & \\
$\mathcal{C}_{\mathrm{QD}3}\llbracket 256, 96, 16 \rrbracket$
& $16$ & $15/15$ & $160/160$ & $15/15$ & $6/6$ & $28\,800$ & $0.38$ &  \\
$\mathcal{C}_{\mathrm{QD}4}\llbracket 256, 130, 8 \rrbracket$
& $16$ & $6/6$ & $33/33$ & $6/6$ & $6/6$ & $4\,608$ & $0.51$ &  \\
\hline
$\mathcal{C}_{\mathrm{HP}1}\llbracket 65, 9, 4 \rrbracket$
& $5.11$ & $2.20$ & $0/0$ & $-/-$ & $4/4$ & $235$ & $0.14$ & $\bullet$ \\
$\mathcal{C}_{\mathrm{HP}2}\llbracket 241, 121, 3 \rrbracket$
& $10.13$ & $2.52$ & $0/0$ & $-/-$ & $4/4$ & $2\,392$ & $0.50$ & \cite{table_QECCs} \\
$\mathcal{C}_{\mathrm{Bic}1}\llbracket 64, 12, 6 \rrbracket$
& $8$ & $3.25$ & $0/0$ & $-/-$ & $4/4$ & $1\,240$ & $0.19$ & $\bullet$ \\
$\mathcal{C}_{\mathrm{Bic}2}\llbracket 64, 18, 2 \rrbracket$
& $8$ & $2.88$ & $0/0$ & $-/-$ & $4/4$ & $1\,048$ & $0.28$ & $\bullet$ \\
$\mathcal{C}_{\mathrm{Bic}3}\llbracket 256, 96, 8 \rrbracket$
& $16$ & $5$ & $0/0$ & $-/-$ & $4/4$ & $18\,112$ & $0.38$ & $\bullet$ \\
$\mathcal{C}_{\mathrm{Bic}4}\llbracket 256, 130, 2 \rrbracket$
& $16$ & $3.94$ & $0/0$ & $-/-$ & $4/4$ & $14\,056$ & $0.51$ & $\bullet$ \\
$\mathcal{C}_{\mathrm{QC}}\llbracket 272, 142, 8 \rrbracket$
& $16$ & $4$ & $3/3$ & $2/2$ & $6/6$ & $2\,176$ & $0.52$ & \cite{baldelli2026XYZ} \\
$\mathcal{C}_{\mathrm{CPM}}\llbracket 276, 98, 14 \rrbracket$
& $12$ & $4$ & $3/3$ & $2/2$ & $6/6$ & $2\,208$ & $0.36$ & \cite{kenta_CPM_2026} \\
$\mathcal{C}_{\mathrm{GB}}\llbracket 48, 6, 8 \rrbracket$
& $8$ & $4$ & $3/3$ & $2/2$ & $4/4$ & $840$ & $0.13$ & \cite{Panteleev_degenerate} \\
$\mathcal{C}_{\mathrm{BB}1}\llbracket 72, 12, 6 \rrbracket$
& $6$ & $3$ & $6/6$ & $3/3$ & $6/6$ & $324$ & $0.17$ & \cite{Bravyi2024_BBcodes_Nature} \\
$\mathcal{C}_{\mathrm{BB}2}\llbracket 288, 12, 18 \rrbracket$
& $6$ & $3$ & $6/6$ & $2/2$ & $6/6$ & $1\,296$ & $0.04$ & \cite{Bravyi2024_BBcodes_Nature} \\
$\mathcal{C}_{\mathrm{D}0}\llbracket 65, 27, 4 \rrbracket$
& $9$ & $3.23$ & $5/5$ & $3/3$ & $4/4$ & $1\,632$ & $0.42$ & $\bullet$ \\
$\mathcal{C}_{\mathrm{D}1}\llbracket 257, 121, 10 \rrbracket$
& $17$ & $7.41$ & $44/44$ & $7/7$ & $4/4$ & $36\,736$ & $0.47$ & \cite{baldelli2026ISIT} \\
\bottomrule
\end{tabular}
}
\end{table}

For comparison, we select codes from state-of-the-art \ac{CSS} \ac{QLDPC}  families with block lengths, rates, and stabilizer generator weights as close as possible to ours, namely: 
bicycle codes~\cite{sparse_quantum_MacKay}; 
\ac{QC} codes~\cite{Hagiwara2007}; 
\ac{HP} codes~\cite{Tillich_HGP};
\ac{GB} codes~\cite{kovalev_Pryad_GBcodes}; 
\ac{BB} codes~\cite{Bravyi2024_BBcodes_Nature};
\ac{CPM}-based codes~\cite{kenta_CPM_2026}; and 
\ac{QD} CAMEL codes~\cite{baldelli2026ISIT}.  In particular, $\mathcal{C}_{\mathrm{D}0}$ and $\mathcal{C}_{\mathrm{D}1}$ are decoded using the CAMEL decoder for the code-capacity model. For the phenomenological model, they are decoded using \ac{BP4}, as for the other codes considered in this study.
The parameters of the benchmark codes are reported, where “$\bullet$” in the “\textbf{Ref.}” column indicates that, to the best of the authors' knowledge, the corresponding code instance has been presented in this work for the first time, while ``--'' denotes the absence of meta-checks ($r_{\mathrm{d}}=0$). In Tab.~\ref{tab:codes_parameters}, all the unknown quantum minimum distances are computed exactly if $n \leq 65$, otherwise they are numerically estimated using~\cite{Pryadko_2022}, which provides an upper bound.  
Similarly, the meta-check distances $d_{\text{m},\mathrm{X}}$ and $d_{\text{m},\mathrm{Z}}$ of the codes with $r_{\mathrm{d,X}} > 0$ and $r_{\mathrm{d,Z}} > 0$ are computed exactly, except for $\mathcal{C}_{\text{QD} 3}$, $\mathcal{C}_{\text{QD} 4}$, and $\mathcal{C}_{\text{D} 1}$. We numerically estimate the meta-check distances of the latter codes using~\cite{MacKayDist}, which yields an upper bound on the minimum distance, as well.
Note that, for the proposed \ac{QD} codes, $d_{\text{m},\mathrm{X}}$ and $d_{\text{m},\mathrm{Z}}$
are equal to $w_{\mathrm{X}}$ and $w_{\mathrm{Z}}$, respectively, thus attaining the theoretical upper bound of Sec.~\ref{subsec:noise_models}.
For each code instance, we use the native pairs of \acp{PCM} without adding or removing any linearly dependent rows.
As shown in~\cite{MiaoQuat}, the initialization of the \ac{BP} decoder strongly affects its performance; thus, for each simulated code, we optimize the initial prior by selecting the value that maximizes the physical error rate at which an \ac{LER} of \(10^{-3}\) is achieved, ensuring a fair comparison\footnote{The \acp{PCM} and the meta-check matrices of all the codes considered in this section, 
as well as the optimized decoder configurations, are available at \url{https://github.com/secomms/High_rate_QD_QLDPC_codes}.}.

\begin{figure*}[t]
    \centering
    \subfloat[]{
        \resizebox{0.9\columnwidth}{!}{
            \input{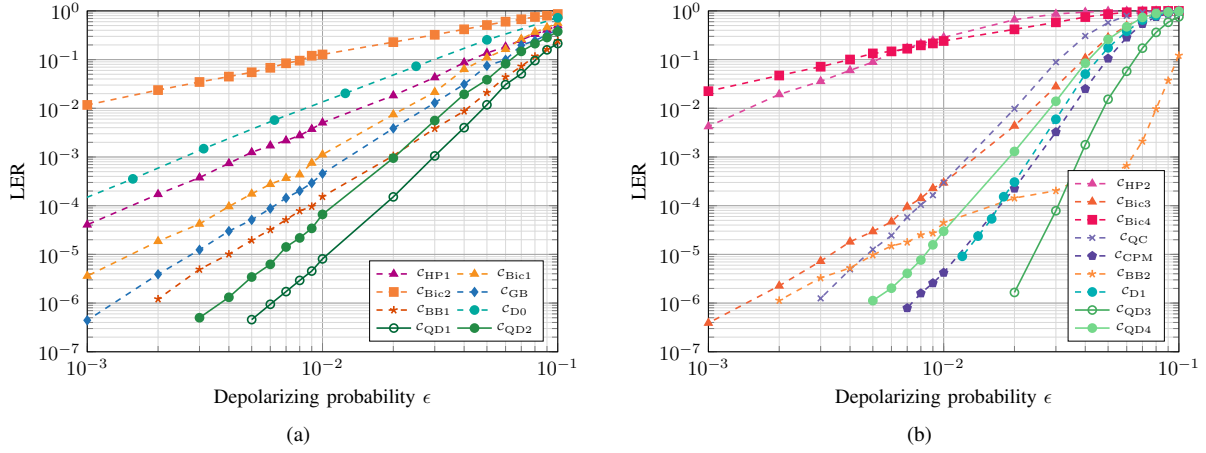}
        }
        \label{fig:BP4_depo_short}
    }
    \subfloat[]{
        \resizebox{0.9\columnwidth}{!}{

\begin{tikzpicture}
    \begin{loglogaxis}[
        xlabel={Depolarizing probability $\epsilon$},
        ylabel={LER},
        xmin=0.001, xmax=0.10,
        ymin=0.0000001, ymax=1,
        legend style={
            legend columns=1, 
            font=\tiny,
        },
        grid=both,
        major grid style={solid, gray!70},
        minor grid style={solid, gray!30},
        legend pos=south east,
        height=7.5cm,
        width=1.1\columnwidth,
        scaled x ticks=false
    ]

    \addplot[
        color=HP2,
        mark=triangle*, thick, dashed,
        mark options={solid}
    ]
    coordinates {
        (0.1, 1)
        (0.09, 1)
        (0.08, 1)
        (0.07, 0.99951171875)
        (0.06, 0.99609375)
        (0.05, 0.98974609375)
        (0.04, 0.9462890625)
        (0.03, 0.865234375)
        (0.02, 0.65966796875)
        (0.01, 0.28515625)
        (0.009, 0.24365234375)
        (0.008, 0.21142578125)
        (0.007, 0.1728515625)
        (0.006, 0.13916015625)
        (0.005, 0.0888671875)
        (0.004, 0.06005859375)
        (0.003, 0.0358664772727)
        (0.002, 0.0192901234568)
        (0.001, 0.00423441734417)
    };

    \addplot[
        color=RedOrange,
        mark=triangle*, thick, dashed,
        mark options={solid}
    ]
    coordinates {
        (0.1, 0.97705078125)
        (0.09, 0.9462890625)
        (0.08, 0.8623046875)
        (0.07, 0.7216796875)
        (0.06, 0.4931640625)
        (0.05, 0.28564453125)
        (0.04, 0.107421875)
        (0.03, 0.0279017857143)
        (0.02, 0.00435236768802)
        (0.01, 0.000295312795313)
        (0.009, 0.00022767011511)
        (0.008, 0.000141749070126)
        (0.007, 9.48464246692e-05)
        (0.006, 4.68150767018e-05)
        (0.005, 2.95485920687e-05)
        (0.004, 1.81422351234e-05)
        (0.003, 7.26072147175e-06)
        (0.002, 2.26201437266e-06)
        (0.001, 3.9e-07)
    };
    
    \addplot[
        color=OrangeRed,
        mark=square*, thick, dashed, 
        mark options={solid}
    ]
    coordinates {
        (0.1, 0.9990234375)
        (0.09, 0.9990234375)
        (0.08, 0.99462890625)
        (0.07, 0.97412109375)
        (0.06, 0.9375)
        (0.05, 0.857421875)
        (0.04, 0.748046875)
        (0.03, 0.5830078125)
        (0.02, 0.41796875)
        (0.01, 0.2412109375)
        (0.009, 0.21484375)
        (0.008, 0.1962890625)
        (0.007, 0.16748046875)
        (0.006, 0.1474609375)
        (0.005, 0.133928571429)
        (0.004, 0.10039893617)
        (0.003, 0.0714962121212)
        (0.002, 0.0471875)
        (0.001, 0.0224702380952)
    };

    \addplot[
        color=QC1,
        mark=x, thick, dashed,
        mark options={solid}
    ]
    coordinates {
        (0.1, 0.99853515625)
        (0.09, 0.99609375)
        (0.08, 0.982421875)
        (0.07, 0.93798828125)
        (0.06, 0.7939453125)
        (0.05, 0.5732421875)
        (0.04, 0.30615234375)
        (0.03, 0.08837890625)
        (0.02, 0.009765625)
        (0.01, 0.000299043062201)
        (0.009, 0.000164197141656)
        (0.008, 0.000103696575524)
        (0.007, 5.80272588851e-05)
        (0.006, 2.44423239371e-05)
        (0.005, 1.25779834977e-05)
        (0.004, 4.91582533955e-06)
        (0.003, 1.24985801613e-06)
    };

    \addplot[
        color=Violet,
        mark=pentagon*, thick, dashed,
        mark options={solid}
    ]
    coordinates {
        (0.1, 0.9677734375)
        (0.09, 0.880859375)
        (0.08, 0.75390625)
        (0.07, 0.541015625)
        (0.06, 0.28515625)
        (0.05, 0.10546875)
        (0.04, 0.0250496031746)
        (0.03, 0.00324844074844)
        (0.02, 0.000225599191452)
        (0.01, 4.2060691222e-06)
        (0.009, 2.56702186035e-06)
        (0.008, 1.59012686414e-06)
        (0.007, 7.9e-07)
    };

    \addplot[
        color=BB3,
        mark=star, thick, dashed,
        mark options={solid}
    ]
    coordinates {
        (0.1, 0.12158203125)
        (0.09, 0.0372023809524)
        (0.08, 0.0097049689441)
        (0.07, 0.00210579514825)
        (0.06, 0.000667164816396)
        (0.05, 0.000413688112258)
        (0.04, 0.000399514190744)
        (0.03, 0.000202974798649)
        (0.02, 0.000143467082912)
        (0.01, 4.48775023695e-05)
        (0.009, 2.73829761133e-05)
        (0.008, 2.51193672331e-05)
        (0.007, 1.78007906399e-05)
        (0.006, 1.49335754564e-05)
        (0.005, 9.69719913858e-06)
        (0.004, 5.26181087115e-06)
        (0.003, 3.28592488533e-06)
        (0.002, 1.13297788645e-06)
    };

    \addplot[
        color = TealBlue,
        mark=*, thick, dashed,
        mark options={solid}
    ]
    coordinates {
        (0.1, 0.972056)
        (0.09, 0.940358)
        (0.08, 0.801047)
        (0.07, 0.649123)
        (0.06, 0.381288)
        (0.05, 0.174797)
        (0.04, 0.0504432)
        (0.03, 0.00594354)
        (0.02, 0.000305356)
        (0.018, 0.00015272)
        (0.016, 5.36207e-05)
        (0.014, 2.37239e-05)
        (0.012, 9.13681e-06)
        };

    \addplot[
        color=QD3,
        mark=o, thick, solid,
        mark options={solid}
    ]
    coordinates {
        (0.1, 0.75732421875)
        (0.09, 0.583984375)
        (0.08, 0.36181640625)
        (0.07, 0.171875)
        (0.06, 0.0576171875)
        (0.05, 0.015318627451)
        (0.04, 0.00178164196123)
        (0.03, 7.86480092616e-05)
        (0.02, 1.65345490063e-06)
    };

    \addplot[
        color=QD4,
        mark=*, thick, solid,
        mark options={solid}
    ]
    coordinates {
        (0.1, 0.97998046875)
        (0.09, 0.9462890625)
        (0.08, 0.88427734375)
        (0.07, 0.7353515625)
        (0.06, 0.47216796875)
        (0.05, 0.26025390625)
        (0.04, 0.08447265625)
        (0.03, 0.0139508928571)
        (0.02, 0.00129883624273)
        (0.01, 2.98255325647e-05)
        (0.009, 1.57114127702e-05)
        (0.008, 7.66115390461e-06)
        (0.007, 4.0747298943e-06)
        (0.006, 2.02345262531e-06)
        (0.005, 1.1230447342e-06)
    };

    \legend{
        {$\mathcal{C}_{\mathrm{HP} 2}$},
        {$\mathcal{C}_{\mathrm{Bic} 3}$},
        {$\mathcal{C}_{\mathrm{Bic} 4}$},
        {$\mathcal{C}_{\mathrm{QC}}$},
        {$\mathcal{C}_{\mathrm{CPM}}$},
        {$\mathcal{C}_{\mathrm{BB} 2}$},
        {$\mathcal{C}_{\mathrm{D} 1}$},
        {$\mathcal{C}_{\mathrm{QD} 3}$},
        {$\mathcal{C}_{\mathrm{QD} 4}$},
    }

    \end{loglogaxis}
\end{tikzpicture}
        }
        \label{fig:BP4_depo_long}
    }
    \caption{
    Comparison between the \ac{LER} of the proposed \ac{QD} \ac{QLDPC} codes and state-of-the-art code instances, as a function of the depolarizing probability $\epsilon$.
    (a): Short block-length codes with $n \approx 64$.
    (b): Moderate block-length codes with $n \approx 256$.}
    
    \label{fig:BP4_depo_long_short}
\end{figure*}

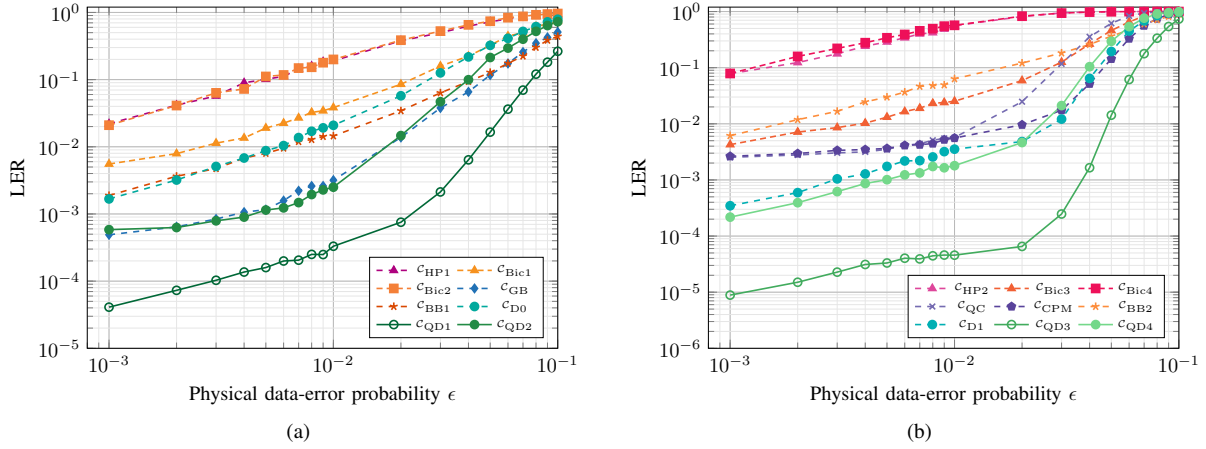
\begin{figure*}[t]
    \centering
    \vspace{-0.275cm}
    \subfloat[]{
        \resizebox{0.9\columnwidth}{!}{
            \begin{tikzpicture}
    \begin{loglogaxis}[
        xlabel={Physical data-error probability $\epsilon$},
        ylabel={LER},
        xmin=8e-4, xmax=0.1,
        ymin=1e-5, ymax=1.2,
        grid=both,
        major grid style={solid, gray!55},
        minor grid style={solid, gray!20},
        height  = 7.5cm,
        width = 1.1\columnwidth, 
        legend style={legend columns=2, font=\tiny},
        legend pos = south east, 
        scaled x ticks=false
    ]

    \addplot[
        color=HP1, mark=triangle*, thick, dashed,
        mark options={solid}
    ]
    coordinates {
        (0.001, 0.0220070422535)
        (0.002, 0.0416666666667)
        (0.003, 0.0571022727273)
        (0.004, 0.0897321428571)
        (0.005, 0.10107421875)
        (0.006, 0.11181640625)
        (0.007, 0.14404296875)
        (0.008, 0.1591796875)
        (0.009, 0.1845703125)
        (0.01, 0.1904296875)
        (0.02, 0.37841796875)
        (0.03, 0.49853515625)
        (0.04, 0.62548828125)
        (0.05, 0.724609375)
        (0.06, 0.8037109375)
        (0.07, 0.84765625)
        (0.08, 0.888671875)
        (0.09, 0.93408203125)
        (0.1, 0.93994140625)
    };

    \addplot[
        color=BurntOrange, 
        mark=triangle*, thick, densely dashed,
        mark options={solid}
    ]
    coordinates {
        (0.001, 0.00553902116402)
        (0.002, 0.00791139240506)
        (0.003, 0.0113379963899)
        (0.004, 0.0135281385281)
        (0.005, 0.0190340909091)
        (0.006, 0.0224820143885)
        (0.007, 0.0269396551724)
        (0.008, 0.0325255102041)
        (0.009, 0.0343406593407)
        (0.01, 0.038300304878)
        (0.02, 0.0853040540541)
        (0.03, 0.1591796875)
        (0.04, 0.22314453125)
        (0.05, 0.32470703125)
        (0.06, 0.44580078125)
        (0.07, 0.49853515625)
        (0.08, 0.61474609375)
        (0.09, 0.67578125)
        (0.1, 0.73779296875)
    };

    \addplot[
        color=Orange, 
        mark=square*, thick, densely dashed,
        mark options={solid}
    ]
    coordinates {
        (0.001, 0.0209731543624)
        (0.002, 0.0411184210526)
        (0.003, 0.0634375)
        (0.004, 0.0726744186047)
        (0.005, 0.111328125)
        (0.006, 0.1171875)
        (0.007, 0.1484375)
        (0.008, 0.15283203125)
        (0.009, 0.1787109375)
        (0.01, 0.19970703125)
        (0.02, 0.38623046875)
        (0.03, 0.52734375)
        (0.04, 0.6533203125)
        (0.05, 0.73974609375)
        (0.06, 0.8359375)
        (0.07, 0.87255859375)
        (0.08, 0.92138671875)
        (0.09, 0.9453125)
        (0.1, 0.962890625)
    };


    \addplot[
        color = GB1, 
        mark=diamond*, thick, solid, dashed,
        mark options={solid}
    ]
    coordinates {
        (0.001, 0.000489811912226)
        (0.002, 0.000643533772652)
        (0.003, 0.000840863453815)
        (0.004, 0.00106256375383)
        (0.005, 0.00115998515219)
        (0.006, 0.00157908034361)
        (0.007, 0.00221160651097)
        (0.008, 0.00260199833472)
        (0.009, 0.00261069340017)
        (0.01, 0.00318228105906)
        (0.02, 0.0137746710526)
        (0.03, 0.0378388554217)
        (0.04, 0.0657552083333)
        (0.05, 0.1171875)
        (0.06, 0.173828125)
        (0.07, 0.2568359375)
        (0.08, 0.34814453125)
        (0.09, 0.4248046875)
        (0.1, 0.509765625)
    };

    \addplot[
        color=BB1, 
        mark=star, thick, solid, dashed,
        mark options={solid}
    ]
    coordinates {
        (0.001, 0.00188366485835)
        (0.002, 0.00364644107351)
        (0.003, 0.00472768532526)
        (0.004, 0.00679347826087)
        (0.005, 0.00797193877551)
        (0.006, 0.00952743902439)
        (0.007, 0.011927480916)
        (0.008, 0.0128073770492)
        (0.009, 0.0142109728507)
        (0.01, 0.0144886363636)
        (0.02, 0.0343406593407)
        (0.03, 0.0634375)
        (0.04, 0.09375)
        (0.05, 0.1279296875)
        (0.06, 0.1748046875)
        (0.07, 0.2236328125)
        (0.08, 0.3017578125)
        (0.09, 0.38525390625)
        (0.1, 0.43994140625)
    };

    \addplot[
        color=Emerald, 
        mark = *, thick, solid, dashed,
        mark options={solid}
    ]
    coordinates {
      (0.001, 0.00167482349166)
      (0.002, 0.0032018442623)
      (0.003, 0.00509289439374)
      (0.004, 0.00680706521739)
      (0.005, 0.00875629194631)
      (0.006, 0.0104166666667)
      (0.007, 0.0137061403509)
      (0.008, 0.0170206971678)
      (0.009, 0.019131097561)
      (0.01, 0.020875)
      (0.02, 0.0574817518248)
      (0.03, 0.126008064516)
      (0.04, 0.217905405405)
      (0.05, 0.32373046875)
      (0.06, 0.41162109375)
      (0.07, 0.521484375)
      (0.08, 0.615234375)
      (0.09, 0.72265625)
      (0.1, 0.79150390625)
    };

    \addplot[
        color=QD1,
        mark=o, thick,
        mark options={solid}
    ]
    coordinates {
        (0.001, 4.1e-05)
        (0.002, 7.3e-05)
        (0.003, 0.000103033300363)
        (0.004, 0.000136617994229)
        (0.005, 0.000159243783123)
        (0.006, 0.000199451110544)
        (0.007, 0.000205227556314)
        (0.008, 0.000250220193771)
        (0.009, 0.000248785924688)
        (0.01, 0.000330687830688)
        (0.02, 0.00075355678804)
        (0.03, 0.00212729748128)
        (0.04, 0.00639059304703)
        (0.05, 0.0165343915344)
        (0.06, 0.0363372093023)
        (0.07, 0.0697916666667)
        (0.08, 0.12060546875)
        (0.09, 0.181640625)
        (0.1, 0.26416015625)
    };

    \addplot[
        color=QD2,
        mark=*, thick,
        mark options={solid}
    ]
    coordinates {
        (0.001, 0.000584002990095)
        (0.002, 0.000628266988339)
        (0.003, 0.000790738866397)
        (0.004, 0.000898504887867)
        (0.005, 0.00114468864469)
        (0.006, 0.00123176980686)
        (0.007, 0.00147684310019)
        (0.008, 0.00194704049844)
        (0.009, 0.0022826880935)
        (0.01, 0.00250400641026)
        (0.02, 0.014675817757)
        (0.03, 0.046875)
        (0.04, 0.099609375)
        (0.05, 0.21337890625)
        (0.06, 0.294921875)
        (0.07, 0.39990234375)
        (0.08, 0.5234375)
        (0.09, 0.63818359375)
        (0.1, 0.7275390625)
    };

    \legend{
        {$\mathcal{C}_{\mathrm{HP} 1}$}, 
        {$\mathcal{C}_{\mathrm{Bic} 1}$},
        {$\mathcal{C}_{\mathrm{Bic} 2}$}, 
        {$\mathcal{C}_{\mathrm{GB}}$}, 
        {$\mathcal{C}_{\mathrm{BB} 1}$}, 
        {$\mathcal{C}_{\mathrm{D} 0}$}, %
        {$\mathcal{C}_{\mathrm{QD} 1}$}, 
        {$\mathcal{C}_{\mathrm{QD} 2}$}, 
    }

    \end{loglogaxis}
\end{tikzpicture}
        }
        \label{fig:BP4_PL_short}
    }
    \subfloat[]{
        \resizebox{0.9\columnwidth}{!}{
            \begin{tikzpicture}
    \begin{loglogaxis}[
        xlabel={Physical data-error probability $\epsilon$},
        ylabel={LER},
        xmin=8e-4, xmax=1.00e-1,
        ymin=1e-6, ymax=1.2,
        grid=both,
        major grid style={solid, gray!55},
        minor grid style={solid, gray!20},
        height=7.5cm,
        width=1.1\columnwidth,
        legend style={legend columns=3, font=\tiny},
        legend pos=south east,
        scaled x ticks=false
    ]

    \addplot[
        color=HP2, mark=triangle*, thick, dashed,
        mark options={solid}
    ]
    coordinates {
        (0.001, 0.078125)
        (0.002, 0.123355263158)
        (0.003, 0.1767578125)
        (0.004, 0.2421875)
        (0.005, 0.28759765625)
        (0.006, 0.3408203125)
        (0.007, 0.39990234375)
        (0.008, 0.42041015625)
        (0.009, 0.4873046875)
        (0.01, 0.560546875)
        (0.02, 0.814453125)
        (0.03, 0.9462890625)
        (0.04, 0.9814453125)
        (0.05, 0.9970703125)
        (0.06, 0.99853515625)
        (0.07, 0.99951171875)
        (0.08, 1)
        (0.09, 1)
        (0.1, 1)
    };

    \addplot[
        color=RedOrange, mark=triangle*, thick, densely dashed,
        mark options={solid}
    ]
    coordinates {
        (0.001, 0.00425749318801)
        (0.002, 0.00713470319635)
        (0.003, 0.00856164383562)
        (0.004, 0.0102809446254)
        (0.005, 0.0131302521008)
        (0.006, 0.0165248691099)
        (0.007, 0.0189393939394)
        (0.008, 0.0231564748201)
        (0.009, 0.0239742366412)
        (0.01, 0.02525)
        (0.02, 0.058837890625)
        (0.03, 0.12548828125)
        (0.04, 0.26416015625)
        (0.05, 0.46044921875)
        (0.06, 0.65234375)
        (0.07, 0.83544921875)
        (0.08, 0.925048828125)
        (0.09, 0.97509765625)
        (0.1, 0.990966796875)
    };
    
    \addplot[
        color=OrangeRed, mark=square*, thick, densely dashed,
        mark options={solid}
    ]
    coordinates {
        (0.001, 0.078857421875)
        (0.002, 0.15869140625)
        (0.003, 0.220458984375)
        (0.004, 0.278564453125)
        (0.005, 0.3408203125)
        (0.006, 0.390625)
        (0.007, 0.453369140625)
        (0.008, 0.495361328125)
        (0.009, 0.54931640625)
        (0.01, 0.56884765625)
        (0.02, 0.8271484375)
        (0.03, 0.939208984375)
        (0.04, 0.97900390625)
        (0.05, 0.992431640625)
        (0.06, 0.998779296875)
        (0.07, 0.998779296875)
        (0.08, 1)
        (0.09, 1)
        (0.1, 1)
    };

    \addplot[
        color=QC1, mark=x, thick, dashed,
        mark options={solid}
    ]
    coordinates {
        (0.001, 0.0025283171521)
        (0.002, 0.00278768956289)
        (0.003, 0.00303398058252)
        (0.004, 0.00315656565657)
        (0.005, 0.00351123595506)
        (0.006, 0.00411184210526)
        (0.007, 0.00442008486563)
        (0.008, 0.00508957654723)
        (0.009, 0.00550176056338)
        (0.01, 0.00569216757741)
        (0.02, 0.0247802734375)
        (0.03, 0.1171875)
        (0.04, 0.35498046875)
        (0.05, 0.6171875)
        (0.06, 0.83642578125)
        (0.07, 0.93994140625)
        (0.08, 0.983642578125)
        (0.09, 0.995361328125)
        (0.1, 0.99853515625)
    };

    \addplot[
        color=Violet, mark=pentagon*, thick, dashed,
        mark options={solid}
    ]
    coordinates {
        (0.001, 0.00263935810811)
        (0.002, 0.0029572740113)
        (0.003, 0.00336255353319)
        (0.004, 0.00347995545657)
        (0.005, 0.00364219114219)
        (0.006, 0.0041335978836)
        (0.007, 0.00424592391304)
        (0.008, 0.00445156695157)
        (0.009, 0.00521694214876)
        (0.01, 0.00557837477798)
        (0.02, 0.00961538461538)
        (0.03, 0.0174479166667)
        (0.04, 0.0517003676471)
        (0.05, 0.1435546875)
        (0.06, 0.329833984375)
        (0.07, 0.560791015625)
        (0.08, 0.748779296875)
        (0.09, 0.884765625)
        (0.1, 0.958251953125)
    };

    \addplot[
        color=BB3, 
        mark=star, thick, dashed,
        mark options={solid}
    ]
    coordinates {
        (0.001, 0.0061394891945)
        (0.002, 0.0118514150943)
        (0.003, 0.0167112299465)
        (0.004, 0.0245881782946)
        (0.005, 0.0299107142857)
        (0.006, 0.0367005813953)
        (0.007, 0.0464348591549)
        (0.008, 0.0482336956522)
        (0.009, 0.0491727941176)
        (0.01, 0.063232421875)
        (0.02, 0.120849609375)
        (0.03, 0.1826171875)
        (0.04, 0.258544921875)
        (0.05, 0.36572265625)
        (0.06, 0.467041015625)
        (0.07, 0.60107421875)
        (0.08, 0.71337890625)
        (0.09, 0.818359375)
        (0.1, 0.89599609375)
    };

    \addplot[        
        color = TealBlue,
        mark=*, thick, dashed,
        mark options={solid}] coordinates {
          (0.001, 0.000346759875721)
          (0.002, 0.000593880653744)
          (0.003, 0.00104445187166)
          (0.004, 0.00127864157119)
          (0.005, 0.00173997772829)
          (0.006, 0.00217921896792)
          (0.007, 0.00221004243281)
          (0.008, 0.00256568144499)
          (0.009, 0.00319529652352)
          (0.01, 0.0035350678733)
          (0.02, 0.0048225308642)
          (0.03, 0.0121394230769)
          (0.04, 0.064453125)
          (0.05, 0.1943359375)
          (0.06, 0.45068359375)
          (0.07, 0.68212890625)
          (0.08, 0.84619140625)
          (0.09, 0.94140625)
          (0.1, 0.97998046875)
    };

    \addplot[
        color=QD3, 
        mark=o, thick, solid,
        mark options={solid}
    ]
    coordinates {
        (0.001, 8.87817720368e-06)
        (0.002, 1.50384020635e-05)
        (0.003, 2.27779640509e-05)
        (0.004, 3.11772530762e-05)
        (0.005, 3.31337649853e-05)
        (0.006, 4.03229275091e-05)
        (0.007, 3.91754558978e-05)
        (0.008, 4.4453611767e-05)
        (0.009, 4.58529380117e-05)
        (0.01, 4.56657703998e-05)
        (0.02, 6.53483152333e-05)
        (0.03, 0.000247165831795)
        (0.04, 0.00164994720169)
        (0.05, 0.0142477203647)
        (0.06, 0.0614853896104)
        (0.07, 0.177734375)
        (0.08, 0.33544921875)
        (0.09, 0.54150390625)
        (0.1, 0.73583984375)
    };

    \addplot[
        color=QD4, 
        mark=*, thick, solid,
        mark options={solid}
    ]
    coordinates {
        (0.001, 0.000217436682438)
        (0.002, 0.000393213656572)
        (0.003, 0.000619220607662)
        (0.004, 0.000864692861096)
        (0.005, 0.00101023706897)
        (0.006, 0.00123128447597)
        (0.007, 0.00133357041252)
        (0.008, 0.00173739807265)
        (0.009, 0.00164820675105)
        (0.01, 0.00179528916124)
        (0.02, 0.00463649851632)
        (0.03, 0.0211148648649)
        (0.04, 0.104166666667)
        (0.05, 0.29541015625)
        (0.06, 0.53955078125)
        (0.07, 0.7626953125)
        (0.08, 0.9072265625)
        (0.09, 0.958984375)
        (0.1, 0.99169921875)
    };

    \legend{
        {$\mathcal{C}_{\mathrm{HP} 2}$},
        {$\mathcal{C}_{\mathrm{Bic} 3}$},
        {$\mathcal{C}_{\mathrm{Bic} 4}$},
        {$\mathcal{C}_{\mathrm{QC}}$},
        {$\mathcal{C}_{\mathrm{CPM}}$},
        {$\mathcal{C}_{\mathrm{BB} 2}$},
        {$\mathcal{C}_{\mathrm{D} 1}$},
        {$\mathcal{C}_{\mathrm{QD} 3}$},
        {$\mathcal{C}_{\mathrm{QD} 4}$},
    }
    
    \end{loglogaxis}
\end{tikzpicture}
        }
        \label{fig:BP4_PL_long}
    }
    \caption{
    Comparison between the \ac{LER} of the proposed \ac{QD} \ac{QLDPC} codes and state-of-the-art code instances, as a function of the physical data-error probability $\epsilon$.  The syndrome readout error rate is set to $p=10^{-2}$. 
    (a): Short block-length codes with $n \approx 64$.
    (b): Moderate block-length codes with $n \approx 256$.
    }
    \label{fig:BP4_PL_long_short}
\end{figure*}

In Fig.~\ref{fig:BP4_depo_long_short}, we show the performance of the four \ac{QD} \ac{QLDPC} codes against state-of-the-art  codes under code-capacity noise.
Specifically, in Fig.~\ref{fig:BP4_depo_long_short}\subref{fig:BP4_depo_short}, we consider the two \ac{QD} codes with $n = 64$. We remark that the performance of $\mathcal{C}_{\text{QD} 1}$ and $\mathcal{C}_{\text{QD} 2}$ is better than that of all the considered benchmark codes.
In particular, for $\mathcal{C}_{\text{HP} 1}$, $\mathcal{C}_{\text{Bic} 1}$, $\mathcal{C}_{\text{Bic} 2}$, $\mathcal{C}_{\text{D} 0}$ and $\mathcal{C}_{\text{BB} 1}$,  this behavior is consistent with their lower minimum distance while, for $\mathcal{C}_{\text{GB}}$, it is consistent with the lower girth of its component \acp{PCM}, equal to $4$.
Moreover, in Fig.~\ref{fig:BP4_depo_long_short}\subref{fig:BP4_depo_long}, we collect the performance of the two  \ac{QD} codes with $n = 256$. 
We observe that $\mathcal{C}_{\text{QD} 3}$ outperforms $\mathcal{C}_{\text{Bic} 3}$ and $\mathcal{C}_{\text{CPM}}$, while $\mathcal{C}_{\text{QD} 4}$ outperforms  $\mathcal{C}_{\text{HP} 2}$ and $\mathcal{C}_{\text{Bic} 4}$, which is consistent with their lower minimum distance.
Although $\mathcal{C}_{\text{QC}}$ and $\mathcal{C}_{\text{QD} 4}$ have the same minimum distance, the latter performs better.
Moreover, considering $\mathcal{C}_{\text{D} 1}$, whose code rate lies between those of $\mathcal{C}_{\text{QD} 3}$ and $\mathcal{C}_{\text{QD} 4}$, we observe that its performance also lies between those of the two codes.
We also include $\mathcal{C}_{\text{BB}2}$ as a benchmark due to its relevance in the literature, despite its significantly lower stabilizer generator weight and code rate.
In particular, $\mathcal{C}_{\text{QD} 3}$ outperforms $\mathcal{C}_{\text{BB}2}$, which exhibits an error floor, for $\epsilon \leq 3.4 \cdot 10^{-2}.$\footnote{In~\cite{Bravyi2024_BBcodes_Nature}, \ac{BB} codes are decoded with a binary \ac{BP} plus ordered-statistics decoder, while here such codes are decoded with a plain \ac{BP4} decoder.}


In Fig.~\ref{fig:BP4_PL_long_short}, we show the performance of the  same codes under phenomenological noise.
The syndrome readout error rate is set to $10^{-2}$.
In Fig.~\ref{fig:BP4_PL_long_short}\subref{fig:BP4_PL_short}, we show the performance of the two \ac{QD} codes with $n = 64$, which outperform again the benchmark codes, consistently with their larger meta-check distance. The only code instances that have similar performance are $\mathcal{C}_{\text{QD} 2}$ and $\mathcal{C}_{\text{GB}}$; however, the \ac{QD} code has more than twice the rate of the \ac{GB} code, namely $0.28$ compared to $0.13$.
In Fig. \ref{fig:BP4_PL_long_short}\subref{fig:BP4_PL_long}, we show the behavior of the two longer \ac{QD} codes, with $n = 256$.
$\mathcal{C}_{\text{QD} 4}$ maintains better performance than $\mathcal{C}_{\text{HP} 2}$, $\mathcal{C}_{\text{Bic} 3}$, $\mathcal{C}_{\text{Bic} 4}$, and $\mathcal{C}_{\text{QC}}$, consistently with the small meta-check distances of these codes.  Notably, under the phenomenological noise model, $\mathcal{C}_{\text{QD}4}$ outperforms both $\mathcal{C}_{\text{CPM}}$ and $\mathcal{C}_{\text{BB}2}$, while also achieving performance competitive with $\mathcal{C}_{\text{D}1}$. This contrasts with the behavior observed under the (less realistic) code-capacity noise model. 
Finally, $\mathcal{C}_{\text{QD} 3}$ outperforms all  benchmark codes, consistently with its larger meta-check distance ($15$).

\section{Conclusions}
\label{sec:concl}

We have introduced a design method for high-rate \ac{CSS} \ac{QD} \ac{QLDPC} codes based on an affine-Frobenius construction.
The resulting  \ac{CSS} codes are characterized by a girth of at least $6$ in the component \acp{PCM}.
We have also derived a closed-form expression for the number of length-$4$ cycles in the quaternary Tanner graph of the designed codes.  Monte Carlo simulations show that the proposed \ac{QD} codes achieve \ac{LER} performance competitive with state-of-the-art \ac{QLDPC} code families in the considered finite-length regimes, under both code-capacity and phenomenological-noise.

\bibliographystyle{IEEEtran}
\bibliography{strings, References_short}

\end{document}